\documentclass[reqno]{amsart}
\usepackage{graphicx}

\newtheorem{theorem}{Theorem}

\newtheorem{corollary}{Corollary}

\newtheorem{remark}{Remark}

\newcommand{\be}{\begin{eqnarray}}
\newcommand{\ee}{\end{eqnarray}}
\newcommand{\bee}{\begin{eqnarray*}}
\newcommand{\eee}{\end{eqnarray*}}
\newcommand{\R}{{\mathbb R}}

\newcommand{\Lcost}{\mbox {\rm L}}
\newcommand{\E}{{\mu}}

\newcommand{\sn}{\mbox {\rm sn}}
\newcommand{\cn}{\mbox {\rm cn}}
\newcommand{\dn}{\mbox {\rm dn}}

\newcommand{\Npsi}{\phi}

\begin{document}
 
 \title [NLS band functions]{Stationary solutions to cubic nonlinear Schr\"odinger equations with quasi-periodic boundary conditions}
 
 \author {
 Andrea Sacchetti
 }

\address {
Department of Physics, Informatics and Mathematics, University of Modena and Reggio Emilia, Modena, Italy.
}


\date {\today}

\thanks {This work is partially supported by GNFM-INdAM and by the UniMoRe-FIM project ``Modelli e metodi della Fisica Matematica''.}

\begin {abstract} 
In this paper we give the \emph {quantization rules} to determine the normalized stationary solutions to the cubic nonlinear 
Schr\"odinger equation with quasi-periodic conditions on a given interval. \ Similarly to what happen in the Floquet's theory for linear periodic operators, also 
in this case some kind of band functions there exist.
\end{abstract}

\maketitle

\section {Introduction}

Nonlinear one-dimensional Schr\"odinger equations with cubic nonlinearity (hereafter NLS)
\be
i \hbar \frac {\partial \psi}{\partial t} = - \frac {\hbar^2}{2m} \frac {\partial^2 \psi}{\partial x^2} + \alpha |\psi |^2 \psi \label {Eq0}
\ee
on a finite interval $I=[0,a]$, for some $a>0$ fixed, may be of physical interest in the study of Bose-Einstein condensates trapped in a circular wave-guide (see, e.g., 
\cite {BR,Gupta,Morizot}, see also \cite {Fibich} for a review). \ Wide attentions to the NLS (\ref {Eq0}) on a finite interval have been given from a mathematical point 
of view, with particular emphasis to the analysis of the existence and stability of \emph {standing waves} solutions of the form 
$\psi (x,t) = e^{-i\mu t /\hbar} \Npsi (x)$  under different boundary conditions \cite {Angulo,Gallay1,Gallay2,LandauWilde,
Robinson,Rowlands}. \ In fact, the function $\Npsi (x)$ is a normalized solution to the cubic time independent NLS (hereafter $'$ 
denotes the derivative $\frac {\partial }{\partial x}$ and, for sake of simplicity, we fix the units such that $\hbar =1$ and $2m=1$):
\be
-\Npsi''  + \alpha |\Npsi |^2 \Npsi = {\E} \Npsi \, , \ \| \Npsi \|_{L^2 (I,dx)} =1  \, ,  \label {Eq1}
\ee
and the boundary conditions considered in the above mentioned papers are the following ones: \emph {periodic boundary conditions} (i.e. $\Npsi (0) = \Npsi (a)$ 
and $\Npsi ' (0)=\Npsi ' (a)$), \emph {Dirichlet boundary conditions} (i.e. $\Npsi (0)=\Npsi (a)=0$), \emph {Neumann boundary conditions} (i.e. $\Npsi ' (0)=
\Npsi '(a)=0$) and \emph {$\sigma$-walls boundary conditions} (i.e. $\Npsi ' (0) =\sigma \Npsi (0)$ and $\Npsi '(a) =-\sigma \Npsi (a)$) 
where the walls are repulsive if $\sigma >0$ and attractive when $\sigma <0$. \ In a couple of seminal papers Carr, Clark and Reinhardt \cite {Carr1,Carr2} studied the 
stationary solutions to (\ref {Eq1}) on a torus, that is with periodic boundary conditions, both in the case of attractive and repulsive cubic nonlinearities,  
 and a key ingredient in their analysis was the use of the fundamental solution to a cubic NLS expressed through elliptic functions (see also \cite {POC,Carr3}). 

This paper is addressed to the study of the normalized stationary solutions $\Npsi (x)$ to equation (\ref {Eq1}) with \emph {quasi-periodic boundary conditions} on the 
interval $I=[0,a]$
\be
\left \{
\begin {array}{lcl}
 \Npsi (a) &=& e^{ika} \Npsi (0) \\ 
 \Npsi' (a) &=& e^{ika} \Npsi' (0)
\end {array}
\right. \, , \ k \in \R \, .
\label {QP}
\ee

In the following, for argument's sake, we choose $a=1$. 

Similarly to what happens in the Floquet's theory for linear periodic operators \cite {Kohn}, even in this case we expect that it is possible to obtain an implicit relationship 
between the ``energy'' $ \E $ associated to the stationary solution and the ``quasimomentum'' variable $k$ that characterizes the quasi-periodic boundary 
conditions. \ Eventually, 
some analogies between the NLS equation (\ref {Eq1}) with quasi-periodic boundary conditions (\ref {QP}) and the Floquet's theory occur; for instance, in additions 
to plane wave solutions associated to the ``energy'' $\E = k^2 + \alpha$, other quasi-periodic solutions there exists for some values of the energy $\E \in [\E^m ,\E^M ]$ 
and of the quasimomentum $k \in [k^m ,k^M]$. \ The intervals $[\E^m , \E^M ]$ and $[k^m ,k^M]$ will depend on $\alpha$ and their amplitude is not zero when 
$\alpha \not=0$. \ Finally, we give the algorithm for the computation of $k=k(\E )$, when the ``energy'' $\E$ belongs to the ``energy band'' $[\E^m , \E^M ]$, and 
the numerical inversion of such a relation gives the ``dispersion relation'' $\E = \E (k)$. \ The names ``energy band'', ``quasimomentum'', ``dispersion relation'', etc., 
are adopted by the Floquet's theory.

The paper is organized as follows. \ In \S \ref {Sec2} we collect some preliminary remarks. \ In \S \ref {Sec3} we give the expression of the general solution to 
(\ref {Eq1}) with quasi-periodic boundary conditions (\ref {QP}) and we compute the ``energy band'' $[\E^m , \E^M ]$ and the associated interval $[k^m , k^M ]$ of 
values for the ``quasimomentum`` $k$ to whom a normalized solution to (\ref {Eq1}) with boundary conditions (\ref {QP}) there exists. \ In particular, we also see 
that when the energy takes a value $\E^m$ or $\E^M$ at the edge of the energy band then we recover well known solutions. \ Finally, in Appendix \ref {AppA} we collect 
some fundamental formulas concerning Jacobian elliptic functions.

\section {Preliminary remarks} \label {Sec2}

\begin {remark}
If $\Npsi (x)$ is a solution to (\ref {Eq1}) and (\ref {QP}) associated to an energy ${\E}$ and to a quasimomentum $k$ then the complex conjugate $\overline {\Npsi (x)}$ 
is still a solution to (\ref {Eq1}) and (\ref {QP}) associated to the same energy ${\E}$ and to the opposite quasimomentum $-k$. \ Therefore, we may restrict our attention 
to the case $k>0$.
\end {remark}

\begin {remark} We recall that if $\Npsi \in H^2 (I)$ is a solution to the differential equation (\ref {Eq1}) when $I=\R$ then 
(see Lemma 3.7 \cite {P}) $\Npsi$ is, up to a phase factor, a real-valued solution. \ We must remark that this regularity result does not hold true when $I=[0,1]$ is a 
finite interval and thus we actually may have complex-valued solutions to equation (\ref {Eq1}) with quasi-periodic boundary conditions (\ref {QP}).
\end {remark}

\begin {remark} \label {Nota1}
Equation (\ref {Eq1}), with quasi-periodic boundary conditions (\ref {QP}), always admits plane wave solutions of the form $\Npsi (x) =e^{\pm i \sqrt {\E - \alpha} x}$ 
where $\E =k^2 +\alpha$.
\end {remark}

\begin {remark} \label {Reale}
When one looks for real valued solutions then equation (\ref {Eq1}) takes the form 
\bee
- \Npsi '' + \alpha \Npsi^3 = {\E} \Npsi \,  , 
\eee
and it has a periodic solution (see Ch. 7, \S 10 \cite {Davis}, see also \cite {AS})
\be
\Npsi (x) = \frac {1}{\sqrt {\alpha}} t \sqrt {\frac {2{\E}}{1+t^2}} \sn \left ( (x-x_0) \sqrt {\frac {{\E}}{1+t^2}}; t \right ) \, , \label {reale}
\ee
where $\sn(x;t)$ is an Jacobian elliptic function with parameters $x_0 \in \R$ and $t \in [0,1)$ and real period $4 K(t)$, where $K(t)$ is the complete first elliptic 
integral. \ Making use of some formulas for $\sn (x;t)$ one can gives other forms to the general solution; e.g., instead of (\ref {reale}) the general solution may be written as 
\be
\Npsi (x) = \sqrt {\frac {-2 {\E} t^2}{\alpha}}  \cn \left ( (x-x_0) \sqrt {\frac {{\E}}{1-2t^2}}; t \right ) \, , \label {realeBis}
\ee
or
\be
\Npsi (x) = \sqrt {\frac {2 {\E}}{\alpha (2-t)}}  \dn \left ( (x-x_0) \sqrt {\frac {{\E}}{t-2}}; t \right ) \, , \label {realeTer}
\ee
where $\cn (x;t)$ and $\dn (x;t)$ are Jacobian elliptic functions.
\end {remark}

\begin {remark}
In the case of periodic boundary conditions then the solution has the form (\ref {reale}) where $\E$ is given by 
\bee
\E  = 16 (n+1)^2 K^2 (t) (1+t^2) \, , \ n=0,1,2,\ldots \, . 
\eee
On the other hand, in the case of out of phase boundary conditions, that is when $k=(2n+1) \pi$, then the solution is still given
by (\ref {reale}) where 
\bee
\E  = 4 (2n +1)^2 K^2 (t) (1+t^2) \, . 
\eee
In both cases the value of the parameter $t$ must be such that the normalization condition holds true.
\end {remark}

\begin {remark} \label {shift}
If $\Npsi (x)$ is a solution to equations (\ref {Eq1}) and (\ref {QP}) then $\Npsi_{x_0} (x) = \Npsi (x-x_0)$ is a solution to equations (\ref {Eq1}) and 
(\ref {QP}), too. \ Indeed, $u(x) = e^{-ikx} \Npsi (x)$ is a periodic function with period $1$ and then 
\bee
e^{-ikx} \Npsi_{x_0} (x) = e^{-ikx_0} e^{-ik(x-x_0)} \Npsi (x-x_0) = e^{-ikx_0} u (x-x_0)
\eee
is a periodic function with period $1$, too.
\end {remark}

\section {Solution to (\ref {Eq1}) with boundary conditions (\ref {QP})} \label {Sec3}

\subsection {Preliminaries}

Following the approach proposed by \cite {Carr1,Carr2} we consider the Madelung transform $\Npsi (x) = \rho (x) e^{i \theta (x)}$, then equation (\ref {Eq1}) 
takes the form  
\be
\left \{ 
\begin {array}{l}
- (\rho'' - \rho {\theta'}^2) + \alpha \rho^3 = {\E} \rho \\ 
2\rho' \theta' +  \rho \theta '' =0 
\end {array}
\right. \, . \label {Eq2}
\ee
The second equation implies that $\rho^2 \theta ' = C_1$, where $C_1 = \rho_0^2 \theta_0'$ is a constant of integration. 

Hereafter we assume, for argument's sake, that 
\bee
\theta_0 = \theta (0) = 0 \, . 
\eee

\begin {remark} 
If $C_1=0$ then $\theta (t) \equiv \theta_0$ and $\Npsi$ is, up to a phase factor, a real valued solution already discussed in Remark \ref {Reale}; indeed, 
equation (\ref {Eq2}) reduces to $-\rho'' + \alpha \rho^3 = \E \rho$.
\end {remark} 

Hereafter, we'll consider the case $C_1 \not= 0$. \ In such a case $\rho (x)$ never takes zero values and 
\be
\theta (x) = C_1 \int_0^x \frac {1}{\rho^2(u)} du  \, , \ C_1 \not= 0 \, . \label {Eqtheta}
\ee
Because of the quasi periodic boundary conditions (\ref {QP}) it follows that $\rho (x)$ 
is a non negative solution to 
\be
- \left (\rho'' - \frac {C_1^2}{\rho^3} \right )  + \alpha \rho^3 = {\E} \rho \label {Eq4}
\ee
with \emph {periodic boundary conditions} 
\be
\rho (1) = \rho (0) \ \mbox { and } \ \rho' (1) = \rho' (0)\, .  \label {Eq8Bis}
\ee
The two parameters ${\E}$ and $C_1$ must satisfy to the condition $\theta (1) =k$, i.e. 
\be
C_1 \int_0^1 \frac {1}{\rho^2 (x;C_1,{\E})} dx =k\, , \label {Eq4BIS}
\ee
because of (\ref {QP}), and to the normalization condition, that is 
\be
\int_0^1 {\rho^2 (x;C_1,{\E})} dx =1 \, .  \label {Eq4TER}
\ee

\begin {remark} \label {Rem1}
If one look for constant solutions to (\ref {Eq4}) then $\rho \equiv 1$, because of the normalization condition (\ref {Eq4TER}), $C_1 =k$, because 
of (\ref {Eq4BIS}), $\theta (x) = k x $ and equation (\ref {Eq4}) reduces to 
\be
 \frac {C_1^2}{\rho^3} + \alpha \rho^3  = {\E} \rho  \, . \label {Eq2Bis}
\ee
Therefore, since $\rho \equiv 1$, it follows that ${\E} = k^2+ \alpha$ and $\Npsi (x)=e^{ik x}$ is the plane wave solution already discussed in Remark \ref {Nota1}.
\end {remark}

In order to look for non constant solutions we remark that equation (\ref {Eq4}) can be solved by means of a simple squaring; indeed, if $\rho$ is not a 
constant function (we have already discussed this case in Remark \ref {Rem1}) then (\ref {Eq4}) reduces to 
\be
-\frac 12 {\rho'}^2 - \frac 12 \frac {C_1^2}{\rho^2} + \frac 14 \alpha \rho^4  - \frac 12 {\E} \rho^2 = C_2 \label {Eq5Bis}
\ee
where $C_2$ is a constant of integration. \ In we set $z=\rho^2$ then $z(x;C_1,C_2,{\E})$ is a non negative solution to the equation  
\be
{z'}^2 = f(z) \ \mbox { where } \ f(z) =  b z^3 + c z^2 + d z + e \label {A3Bis}
\ee
with periodic boundary conditions $z(0)=z(1)$, where  
\bee
b = 2 \alpha \, , \ c = - 4 {\E} \, , \ d = - 8 C_2 \ \mbox { and } \ e = -4 C_1^2.
\eee
It is well known that the general solution to (\ref {A3Bis}) has the form \cite {Davis}
\be
z (x) = A \sn^2 (q x+x_0 ;t) + B \, \ \mbox { with period } \ T_\ell =\frac {2\ell K(t)}{q} \, , \ \ell =1,2, \ldots \, , \label {EqSol}
\ee
for some $A$, $B$, $q$, $x_0$ and $t$, under the constraints
\bee
C_1^2 > 0 \, ,\ B >0 \ \mbox { and } \ A>-B \, ,
\eee
because we assumed that $C_1 \not= 0$ and that $z(x)$ never takes zero values.

For argument's sake we can always assume that (see Remark \ref {shift}) 
\bee
x_0=0
\eee
by means of a translation argument $x\to x -x_0/q$. 

\begin {remark} \label {Prima}
In the following we restrict our attention to the case of $\ell =1$ in (\ref {EqSol}) and we denote by 
\bee
T :=T_1 = \frac {2K(t)}{q} 
\eee
the corresponding period of the solution $z(x)$. \ In such a way we'll obtain the first ``band function'' $\E := \E_1 = \E_1 (k)$. \ For 
different values of $\ell =2,3, \ldots $ we'll have the other ``band functions'' $\E_\ell (k)$.
\end {remark}

\begin {remark}
One can obtain a general solution to equation (\ref {A3Bis}) even when $ f(z) =  az^4+b z^3 + c z^2 + d z + e$ is a fourth degree polynomial with $a\not= 0$; this case corresponds 
to the cubic/quintic NLS
\bee
-\Npsi'' + \alpha |\Npsi|^2 \Npsi + \beta |\Npsi|^4 \Npsi = {\E} \Npsi \, ,
\eee 
where $a= \frac 43 \beta$. \ Indeed, let any $z_0 >0$ be fixed; then it is known that when $f(z)$ is a quartic polynomial with non repeating factors then 
equation (\ref {A3Bis}) has a general solution given by $z(x) = \zeta (\pm x)$ where
\be
\zeta (x) = z_0 + \frac {\sqrt {f(z_0)} P' (x) + \frac 12 \dot f (z_0) \left [ P(x) - \frac {1}{24} 
\ddot f (z_0) \right ] + \frac {1}{24} f(z_0) f^{(3)} (z_0)}{2\left [ P(x) - \frac {1}{24} 
\ddot f (z_0) \right ]^2 - \frac {1}{48} f(z_0) f^{(IV)} (z_0)}
\label {A5}
\ee
where $' = \frac {d}{dx}$ denotes the derivative with respect to $x$ and $\dot {} =\frac {d}{dz}$ denotes the derivative with respect to $z$, and 
where $P(x) = P(x;g_2,g_3)$ is the Weierstrass's elliptic function with parameters 
\bee 
g_2 = ae - \frac 14 bd +\frac {1}{12} c^2 \mbox { and } \ g_3 = - \frac {1}{16} eb^2 + \frac 16 eac - \frac {1}{16} ad^2 + \frac {1}{48} dbc -
\frac {1}{216}c^3. 
\eee
This is an old and, as far as I know, almost unknown result due to Weierstrass. \ It was
 published in 1865, in an inaugural dissertation at Berlin, by Biermann \cite {Bie}, who ascribed it to Weierstrass; it was then mentioned in the book 
by Whittaker and Watson \cite {WW} in Ch. XX, Example 2, p.454. \ When $\beta =0$ then it is possible to prove that (\ref {A5}) reduces to (\ref {EqSol}).
\end {remark}

Recalling (\ref {Eq8Bis}), that is the period $T$ of the solution $\rho (x)$ must be equal to $1$, the normalization condition 
\bee
1 = \int_0^1 z(x) dx = \int_0^1 \left [ A \sn^2 (qx;t) + B \right ] dx = AF_1(t) + B \, , 
\eee
where we set 
\bee
F_1(t):= \int_0^1  {\sn}^2 (qx;t) dx \, , 
\eee
and by substituting (\ref {EqSol}) in (\ref {A3Bis}) and equating the coefficients of the same power of the function $ \sn^2 (qx ;t)$, it follows that
\be
q & = & 2 K(t) \label {q1} \\
A & = & \frac {1}{\alpha} 2 q^2 t^2 =  \frac {8}{\alpha} K^2(t) t^2 \label {q2} \\ 
B & = & 1- \frac {8K^2(t) t^2 F_1(t)}{\alpha} \label {q3} \\ 
{\E}  & = &  G(t) + \frac 32 \alpha \label {q6} \\ 
C_1^2 & = & \frac {B}{4} (A+B)(2\alpha B + 4 q^2)  \label {q5} 
\ee
where we set
\be
G(t):=4 K^2(t) \left [ \left ( 1 + t^2 \right ) -3 t^2  F_1(t) \right ] \, . \label {pippo}
\ee
Finally, the constant of integration $C_2$ is given by
\bee
C_2  =  -\frac 12 \alpha A B - B q^2 - \frac 34 \alpha B^2 - \frac 12 A q^2 \, . 
\eee

 \begin {remark} \label {Inv}
 We may remark that $G(t)$ is a monotone decreasing function such that
 \bee
 \lim_{t\to 0 } G(t)= 4 K^2 (0) = \pi^2 \ \mbox { and } \  \lim_{t\to 1 } G(t) =- \infty \, . 
 \eee
 \end {remark}

For an explicit formula of the term $F_1 (t)$ we refer to formula (\ref {App1}) in Appendix. 

\subsection {General solution}

From (\ref {q6}) we obtain the equation ${\E} = {\E} (t)$, because of Remark \ref {Inv} we may invert 
such an equation obtaining $t=t({\E} ) \in [0,1)$ and finally $k=k({\E} )$; the inversion of such a latter relation will give the first (because we chose $\ell =1$, see 
Remark \ref {Prima}) ``band function'' $\E (k)$. 
 
 We collect all these results in the following statement.
 
 \begin {theorem} \label {Teo1}
 Let ${\E}\in \R$ and $\alpha \in \R$, $\alpha \not= 0$, be fixed. \ Let $q$, $A$, $B$ and $C_1$ given by (\ref {q1}), (\ref {q2}), (\ref {q3}) and (\ref {q5}). \ Let 
 $t \in [0,1)$ be a solution to the following ``quantization rule''
 \be
 \left \{
 \begin {array}{lcl}
 {\E} &=& G(t) + \frac 32 \alpha \\ 
 B(t) &>& 0 \\
 A (t) &>& - B (t) \\ 
 C_1^2(t) &> & 0 
 \end {array}
 \right. \, . \label {Eq21Bis}
 \ee
 Then  equation (\ref {Eq1}) with quasi-periodic boundary conditions (\ref {QP}) has a solution of the form $\Npsi (x) = \rho (x) e^{i\theta (x)}$ where $\rho (x)$ is 
 a positive  function given by 
 \bee
 \rho (x) = \sqrt {A \sn^2 (qx;t) + B} 
 \eee
 and where 
 \bee
 \theta (x) = C_1 \int_0^x \frac {du}{\rho^2 (u)} \, . 
 \eee
 \end {theorem}
 
 Theorem \ref {Teo1} provides some restrictions to the values allowed by energy ${\E}$. \ Now, we are going to see how these constraints work in the case of 
 attractive and repulsive nonlinearities.
 
 \begin {theorem} \label {Teo2}
 For any $\alpha \in \R$, $\alpha \not=0$, then (\ref {Eq21Bis}) has just one solution $t \in [0,1)$ for any value ${\E}\in ({\E}^m, {\E}^M)$ where 
\bee
{\E}^m= G(t^m) + \frac 32 \alpha \ \mbox { and } \ {\E}^M =  G(t^M) + \frac 32 \alpha \, , 
\eee
and where $t^m$ and $t^M$ are given by (\ref {tm_neg}) in the case of attractive nonlinearity $\alpha <0$, and by (\ref {tm_pos}) in 
the case of repulsive nonlinearity $\alpha >0$.
\end {theorem}

\begin {proof}
Let us consider, at first, the case of attractive nonlinearity, i.e. $\alpha <0$. \ In such a case $B>0$ is always satisfied; furthermore condition $A>-B$ implies 
that 
\bee
A(1-F_1) >-1 \ \mbox { that is } \  8 K^2 (t) t^2 F_2 (t) + \alpha < 0\, ,
 \eee
 where
\bee
F_2(t):= 1- F_1 (t) = \int_0^1  {\cn}^2 (qx;t) dx \, . 
\eee
 Condition $C_1^2 > 0$, under the constraints $B>0$ and $A+B>0$, becomes 
 \bee
 (2\alpha B + 4 q^2 )= 2\alpha + 16 K^2 (t) \left ( 1- t^2 F_1 (t) \right ) > 0 \, . 
 \eee
 In conclusion, when $\alpha <0$ the quantization rule reads as
 \be
 \left \{
 \begin {array}{l}
 {\E} = G(t) + \frac 32 \alpha \\ 
8 K^2 (t) t^2 F_2 (t) + \alpha < 0 \\ 
\alpha + 8 K^2 (t) \left ( 1- t^2 F_1 (t) \right ) > 0 
 \end {array}
 \right. \, . \label {att}
 \ee
 The two functions 
\bee
8 K^2 (t) t^2 F_2 (t) \ \mbox { and } \ 8 K^2 (t) \left ( 1- t^2 F_1 (t) \right )
\eee
are monotone increasing functions such that 
\bee
\lim_{t\to 0} 8 K^2 (t) t^2 F_2 (t) = 0 \ \mbox { and } \ \lim_{t\to 1} 8 K^2 (t) t^2 F_2 (t) =+\infty 
\eee
and 
\bee
\Lcost := \lim_{t\to 0} 8 K^2 (t) \left ( 1- t^2 F_1 (t) \right ) =  2\pi^2 
\eee 
and 
\bee 
\lim_{t\to 1} 8 K^2 (t) \left ( 1- t^2 F_1 (t) \right ) =+\infty \, . 
\eee
Furthermore 
\bee
 8 K^2 (t) \left ( 1- t^2 F_1 (t) \right )-8 K^2 (t) t^2 F_2 (t) = 8K^2(t)(1-t^2) >0 \, , \forall t \in [0,1) \, . 
\eee
In conclusion: let $t_1$ be the unique solution to the equation 
\bee
8 K^2 (t_1) \left ( 1- t_1^2 F_1 (t_1) \right ) = - \alpha \, , 
\eee
and let $t_2$ be the unique solution to the equation 
\bee
8 K^2(t_2) t_2^2 F_2 (t_2) = - \alpha \, . 
\eee
Then, 
\be
t^M = 
\left \{
\begin {array}{ll}
0 & \mbox { if } -\Lcost \le \alpha < 0  \\
t_1 & \mbox { if } \alpha < -\Lcost   
\end {array}
\right. \ \mbox { and } \ t^m =t_2  \label {tm_neg}
\ee
proving thus Theorem \ref {Teo2} in the attractive case. 

We consider now the case of repulsive nonlinearity, i.e. $\alpha >0$. \ Then, condition $B>0$ implies that 
\be
K^2 (t) t^2 F_1(t) < \frac 18  \alpha \, . \label {alphaneg}
\ee
Furthermore, condition $A>-B$ reduces to 
\bee
K^2 (t) t^2 F_2 (t) > - \frac 18 \alpha \, , 
\eee
which is always satisfied. \ Condition $C_1^2 > 0$ becomes 
\bee
\alpha + 8 K^2(t) \left ( 1- t^2 F_1(t) \right ) > 0 
\eee
which is always satisfied, too. \ In conclusion, when $\alpha >0$ the quantization rule reads as
\bee
\left \{
\begin {array}{l}
{\E} = G(t) + \frac 32 \alpha \\ 
\alpha - 8K^2 (t) t^2 F_1(t) > 0 
\end {array}
\right. \, . 
\eee
We remark that the function $8K^2 (t) t^2 F_1(t)$ is a monotone increasing function such that 
\bee
\lim_{t \to 0} 8K^2 (t) t^2 F_1(t) = 0 \ \mbox { and } \ \lim_{t \to 1} 8K^2 (t) t^2 F_1(t) = +\infty \, . 
\eee
Let $t_3$ be the unique solution to the equation
\bee
8K^2 (t_3) t_3^2 F_1(t_3) =  \alpha \, , 
\eee
then 
\be
t^M =0 \ \mbox { and } \ t^m = t_3 \label {tm_pos}
\ee
completing so the proof of the Theorem \ref {Teo2}
\end {proof}
 
The allowed values for the energy ${\E}$ and for the quasimomentum $k$, as function of the nonlinearity parameter $\alpha$, are displaced in Figures \ref {Fig1} 
and \ref {Fig2}. \ In particular we plot the graph of the $\alpha$-dependent functions $\E^m$ and $\E^M $, and the graph of the functions 
\bee
k^m = \inf_{\E \in (\E^m , \E^M )} k(\E ) \ \mbox { and } k^M = \sup_{\E \in (\E^m , \E^M )} k(\E )\, . 
\eee
\begin{center}
\begin{figure}
\includegraphics[height=6cm,width=6cm]{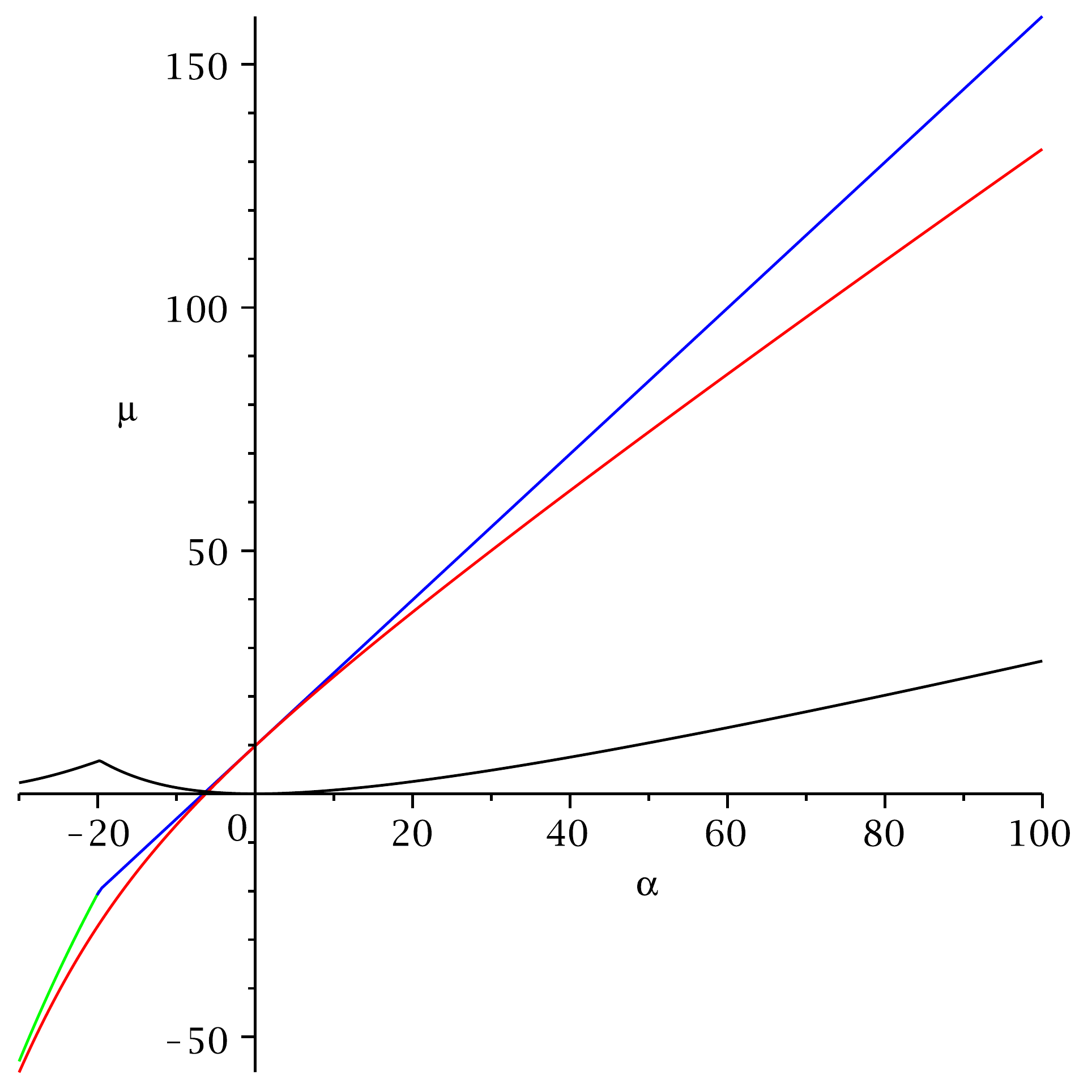}
\caption{\label {Fig1} We consider the cubic model for $\alpha \in [-30,+100]$. \ We plot the graph of the functions ${\E}^M (\alpha )$ 
(green line for $\alpha < - \Lcost$ and blue line for $\alpha >-\Lcost$) and ${\E}^m (\alpha )$ 
(red line). \ The allowed values for the energy ${\E}$ are the ones contained in the interval $(\E^m , \E^M )$. \ If we call $\E^M - \E^m$ the 
band width then it depends on $\alpha$ and it is zero only when $\alpha =0 $ (black line). \ At the edges of the band $(\E^m , \E^M )$ we have that $C_1=0$ 
along the green line; $C_1 \not=0$ and $A=0$ and $B=1$ along the blue line; and finally $A=-B$ along the red line for $\alpha <0$ and $b=0$ along the red line for $\alpha >0$.}
\end{figure}
\end{center}
\begin{center}
\begin{figure}
\includegraphics[height=6cm,width=6cm]{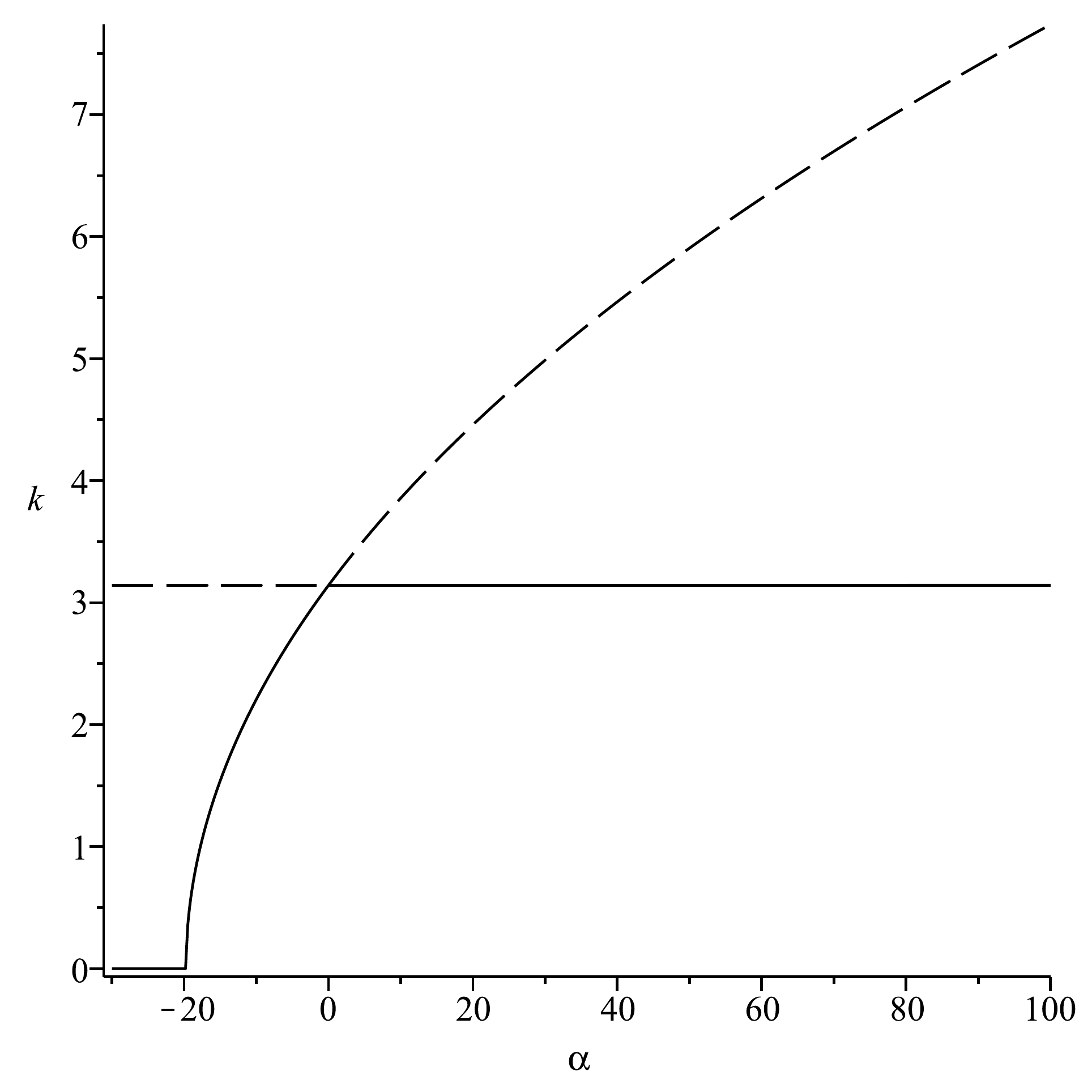}
\caption{\label {Fig2} We consider the cubic model for $\alpha \in [-30,+100]$. \ We plot the graph of the functions $k^m $ (full line) and $k^M$ 
(broken line). \ As proved in Theorem \ref {Teo3} $\lim_{\E \to \E^m} k (\E ) =\pi$.}
\end{figure}
\end{center}

There are three different behavior of the ``band function`` $\E (k)$ and of the associated solutions (see Figure \ref {Fig3}): when $\alpha < - \Lcost$ then the 
band function is defined for any $k \in (0,\pi )$; when $-\Lcost < \alpha <0$ then the band function is defined for any $k \in (k^m ,\pi )$ where $k^m \in (0, \pi)$; 
when $0 < \alpha $ then the band function is defined for any $k \in (\pi , k^M )$ where $k^M  >\pi$.

\begin{center}
\begin{figure}
\includegraphics[height=4cm,width=4cm]{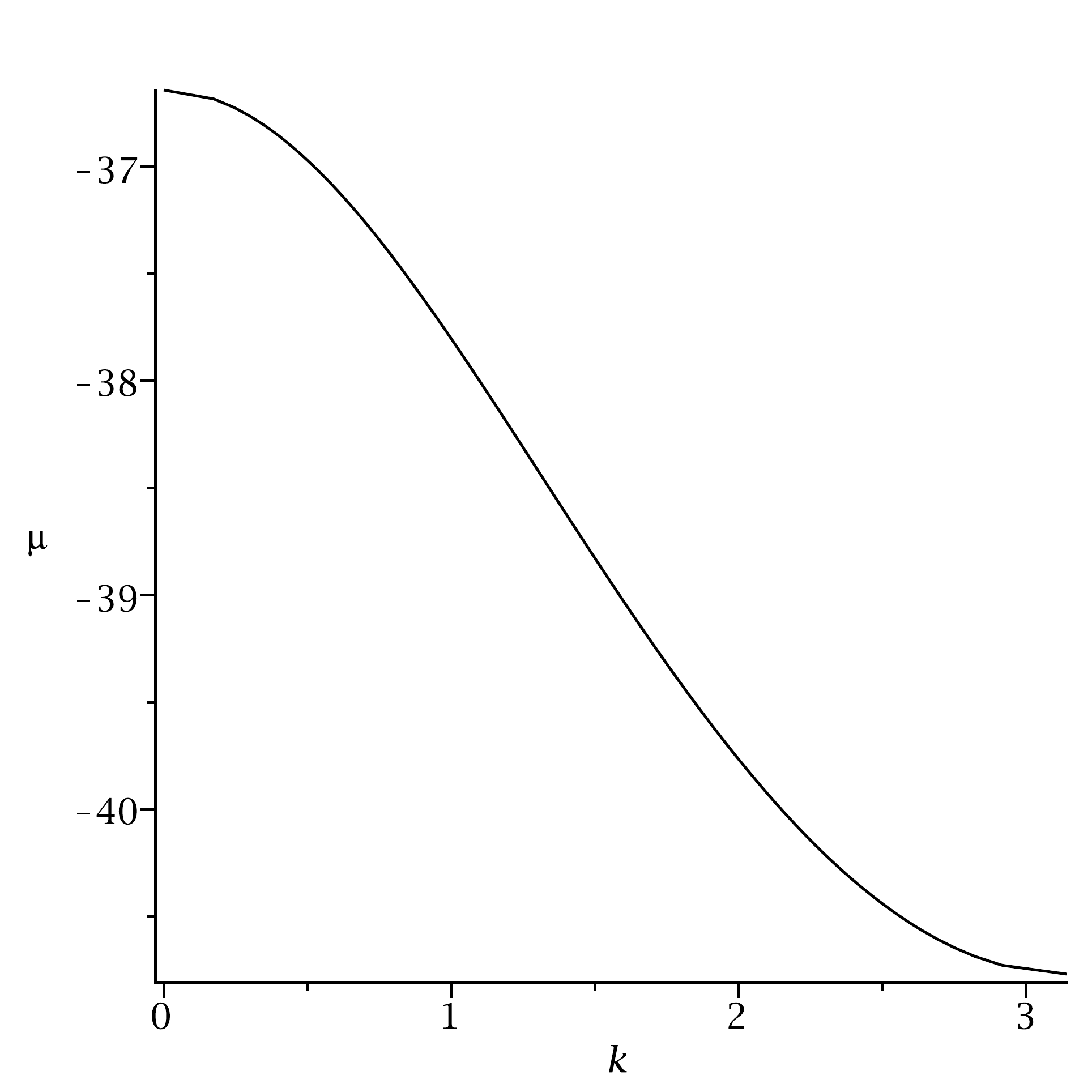}
\includegraphics[height=4cm,width=4cm]{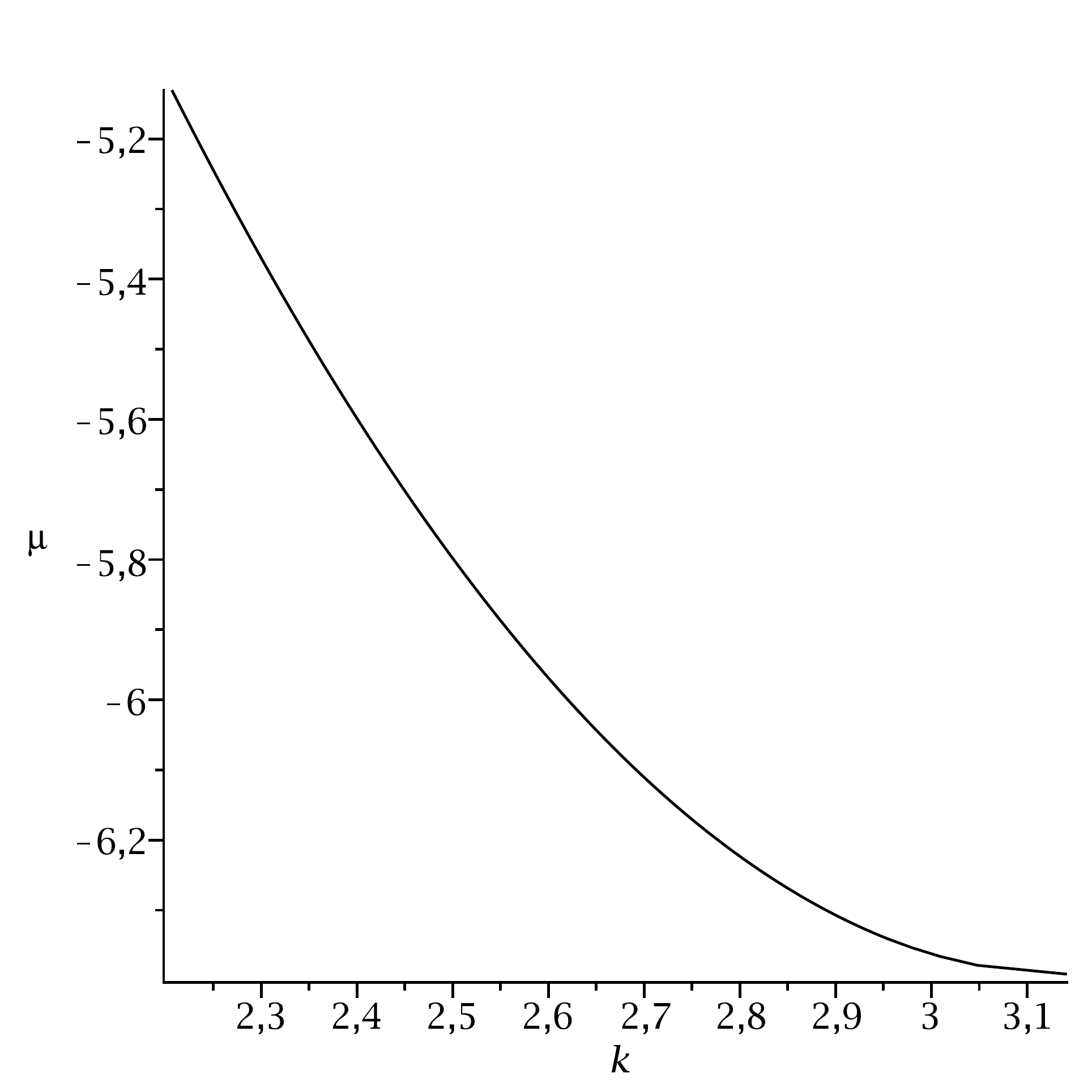}
\includegraphics[height=4cm,width=4cm]{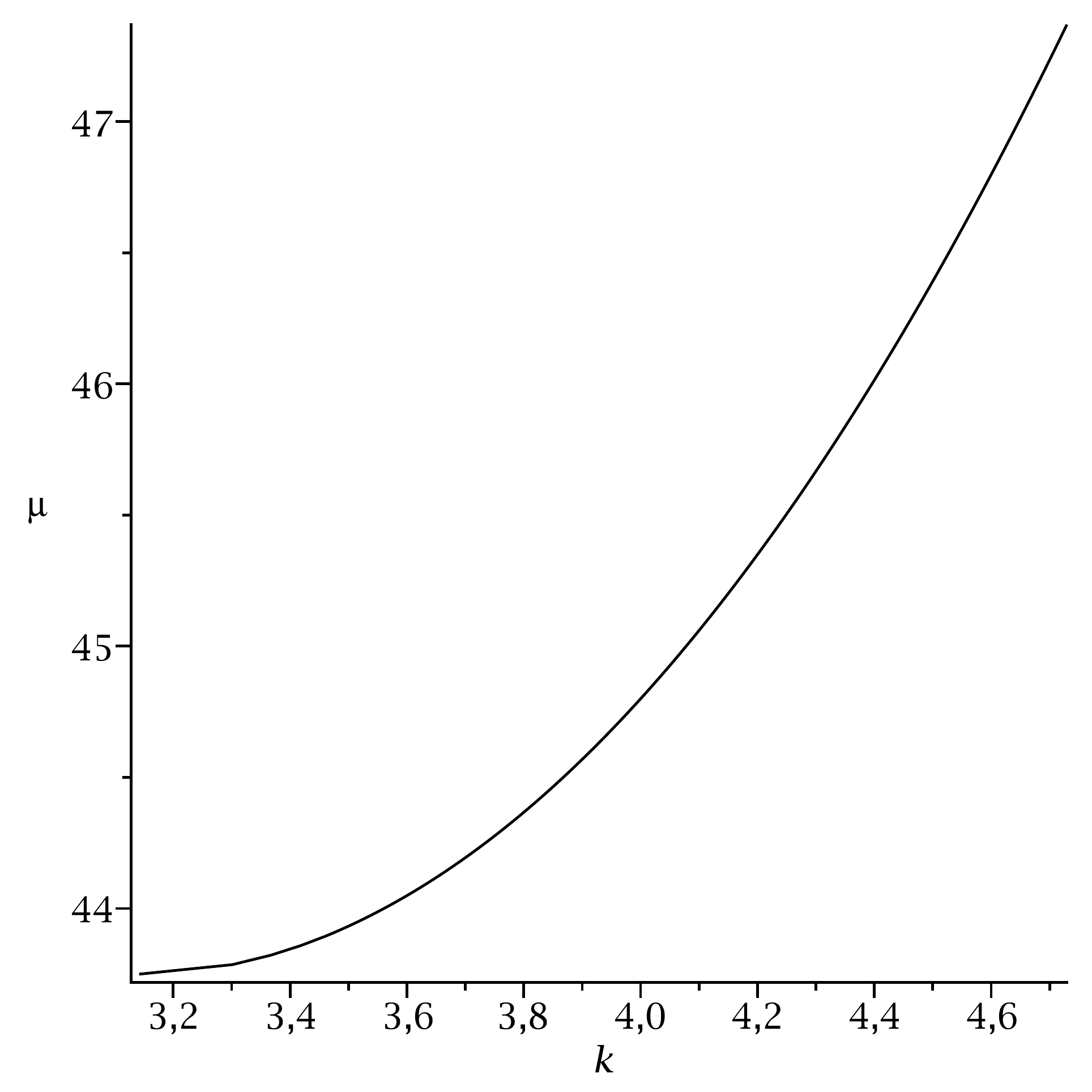}
\caption{\label {Fig3} Here we plot the graph of the band functions $\E (k)$ when $\alpha=-25$ (left hand side panel), $\alpha =-10$ (central side panel) and 
$\alpha =+25$ (right hand side panel).}
\end{figure}
\end{center}
\begin{center}
\begin{figure}
\includegraphics[height=6cm,width=6cm]{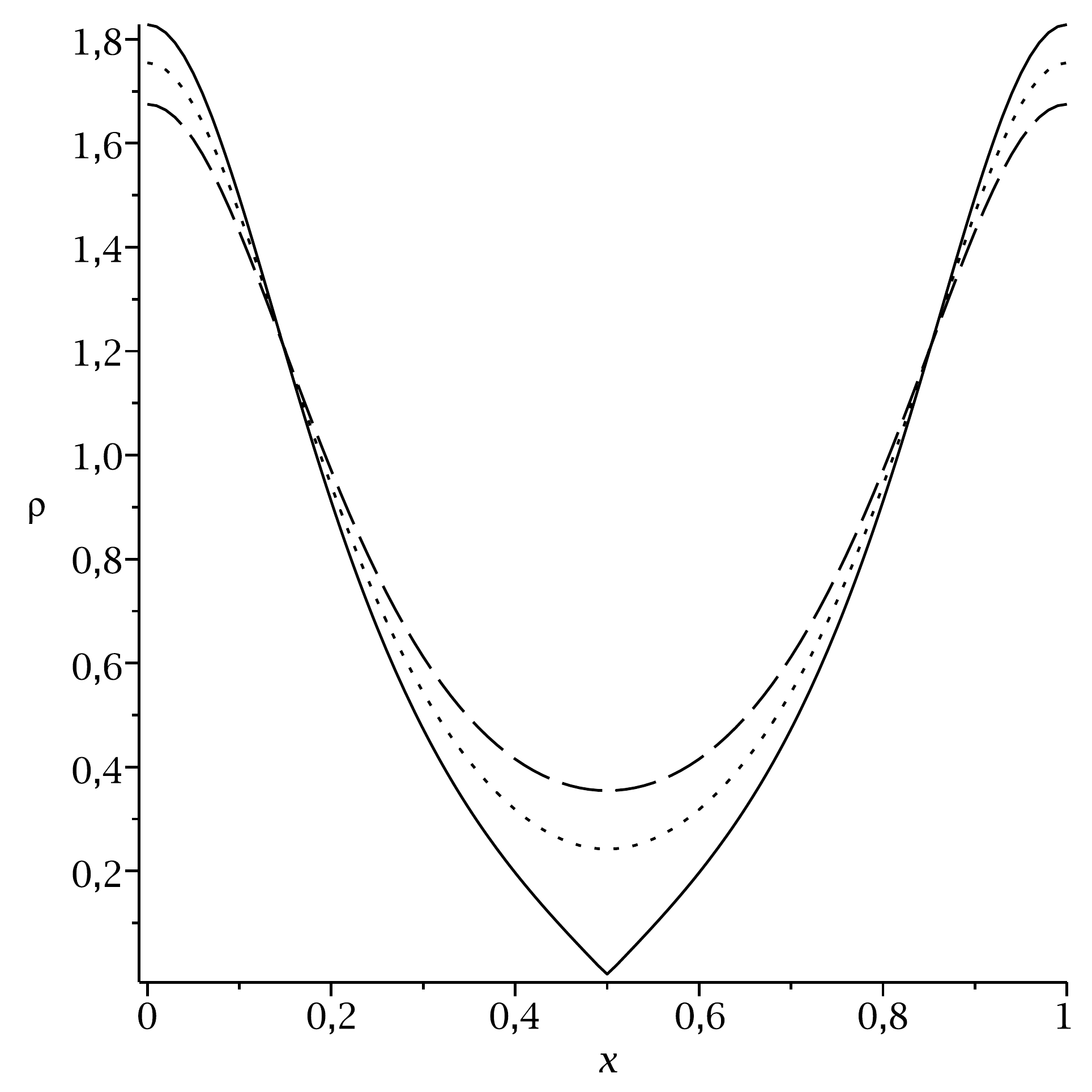}
\includegraphics[height=6cm,width=6cm]{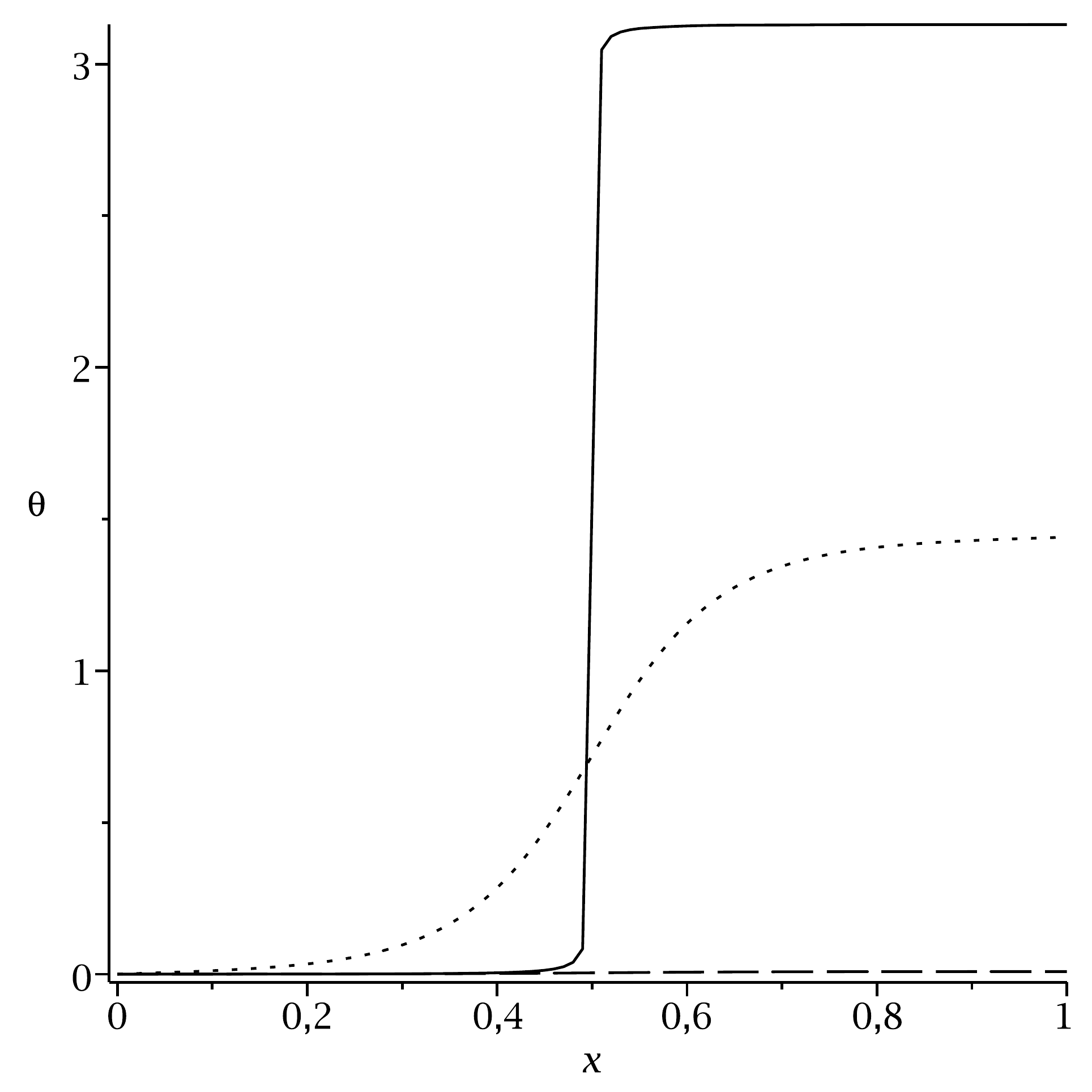}
\caption{\label {Fig9} Here we plot the graph of the solution $\rho (x)$ (left hand side) and $\theta (x)$ (right hand side) when $\alpha=-25$. \ Broken lines correspond 
to the solutions associated to an energy $\E$ close to the value $\E^M$; full lines correspond 
to the solutions associated to an energy $\E$ close to the value $\E^m$; dot lines correspond 
to the solutions associated to the energy $\E = \frac 12 (\E^m + \E^M)$.}
\end{figure}
\end{center}
\begin{center}
\begin{figure}
\includegraphics[height=6cm,width=6cm]{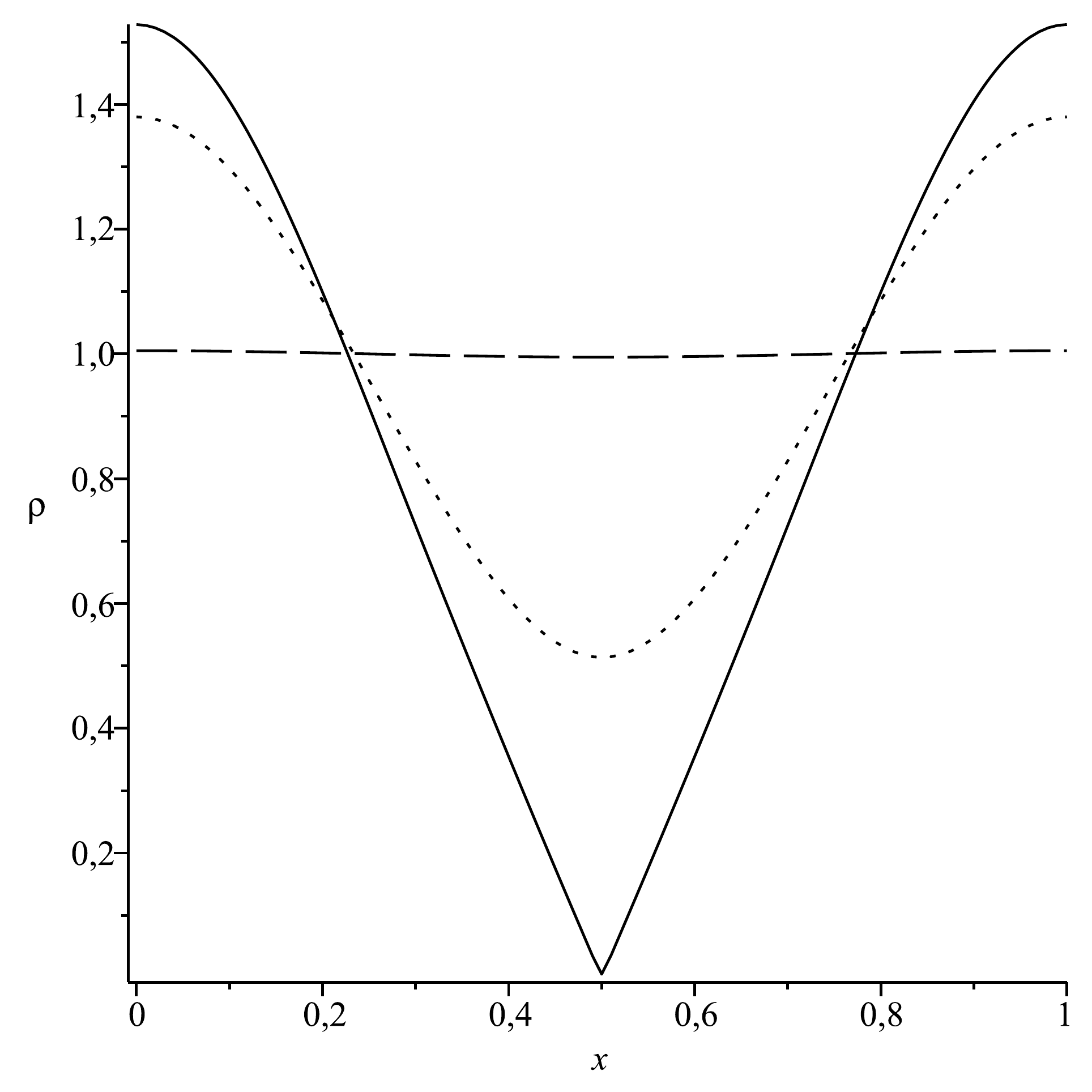}
\includegraphics[height=6cm,width=6cm]{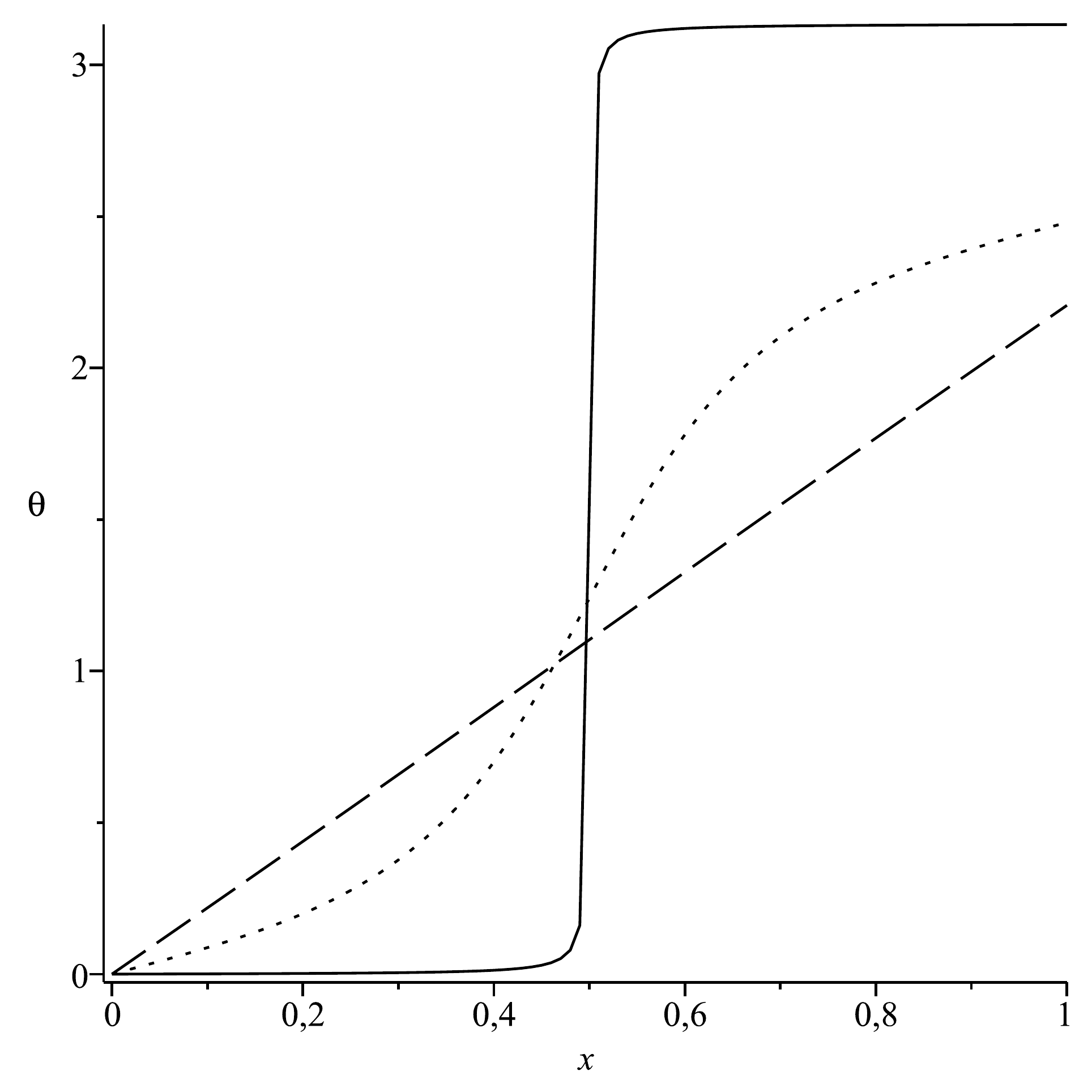}
\caption{\label {Fig10} Here we plot the graph of the solution $\rho (x)$ (left hand side) and $\theta (x)$ (right hand side) when $\alpha=-10$. \ Full lines correspond 
to the solutions associated to an energy $\E$ close to the value $\E^m$; broken lines correspond 
to the solutions associated to an energy $\E$ close to the value $\E^M$; dot lines correspond 
to the solutions associated to the energy $\E = \frac 12 (\E^m + \E^M)$.}
\end{figure}
\end{center}
\begin{center}
\begin{figure}
\includegraphics[height=6cm,width=6cm]{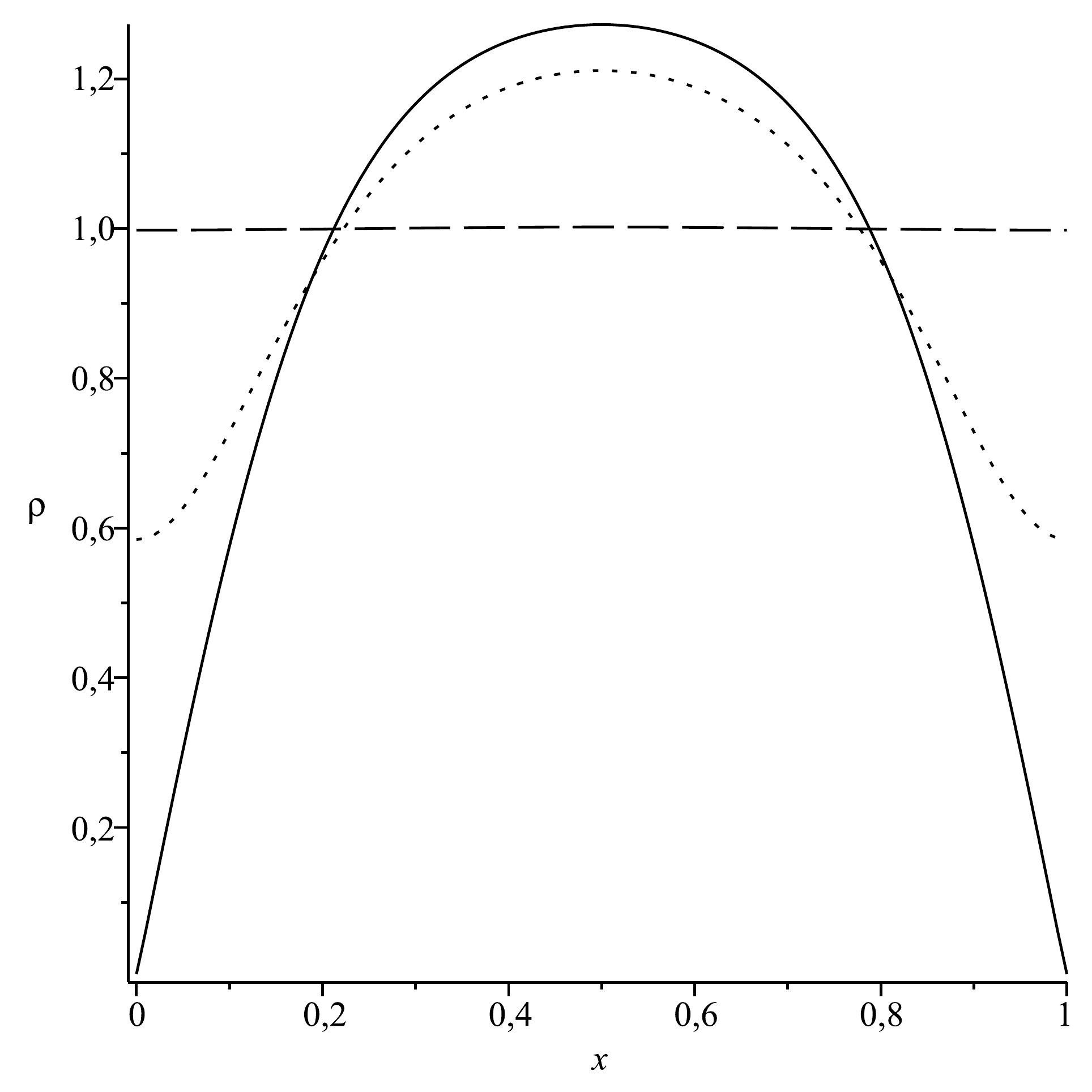}
\includegraphics[height=6cm,width=6cm]{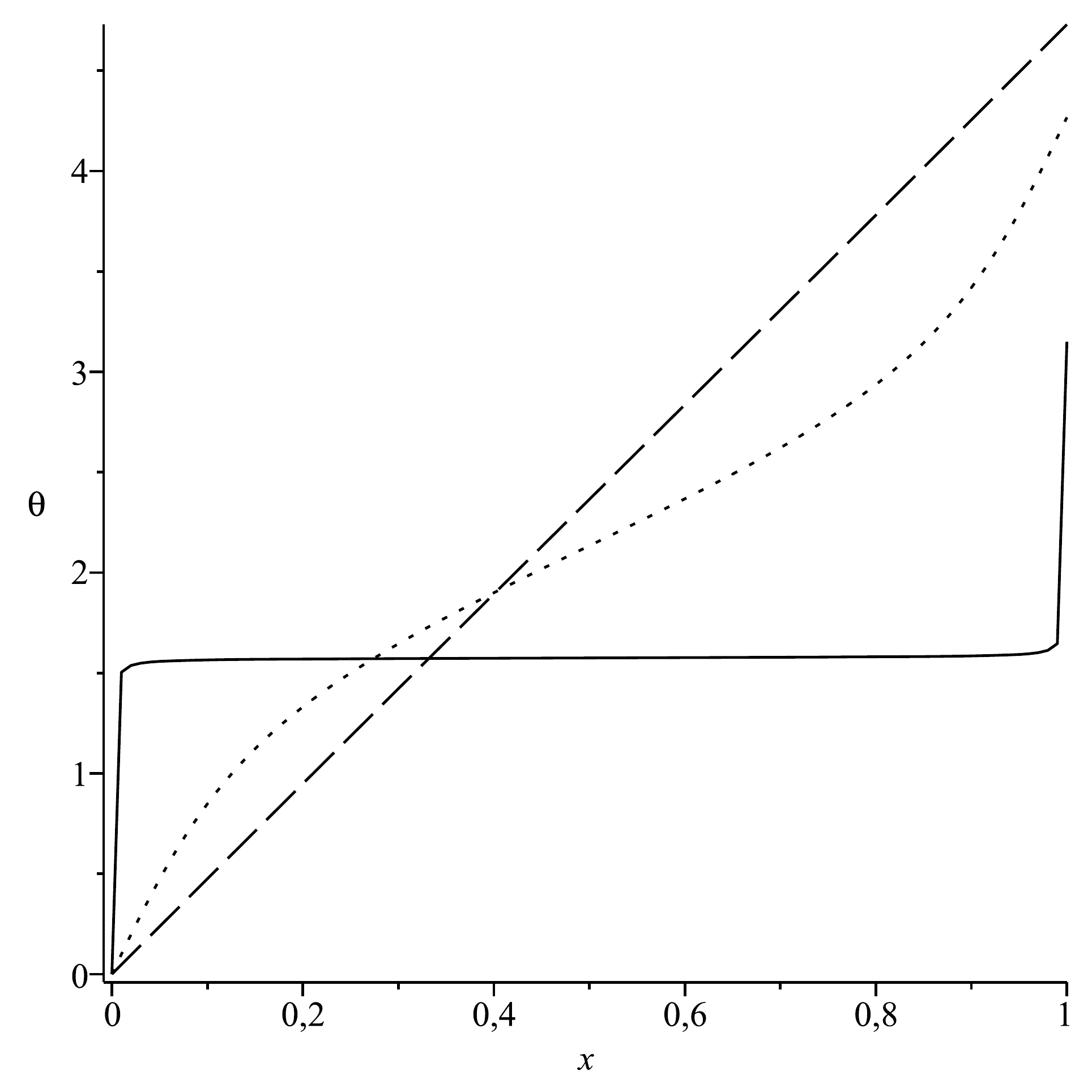}
\caption{\label {Fig11}Here we plot the graph of the solution $\rho (x)$ (left hand side) and $\theta (x)$ (right hand side) when $\alpha= +25$. \ Full lines correspond 
to the solutions associated to an energy $\E$ close to the value $\E^m$; broken lines correspond 
to the solutions associated to an energy $\E$ close to the value $\E^M$; dot lines correspond 
to the solutions associated to the energy $\E = \frac 12 (\E^m + \E^M)$.}
\end{figure}
\end{center}

\subsection {Behavior of the solution at the boundaries $\E^m$ and $\E^M$}

Here, we consider the behavior of the solution $\Npsi$ when $\E$ takes the boundary values $\E^m$ and $\E^M$. 

At first we consider the limit ${\E}\to {\E}^M$ where the proof of the Corollary below is an immediate conseguence of Theorem \ref {Teo2}. 

\begin {corollary} \label {Cor1}
If $\alpha \ge - \Lcost$, then $t^M \to 0$, $A\to 0$, $B \to  1$ and $k \to \sqrt {\alpha/2 
+ q^2 (0)}$ as ${\E}\to {\E}^M=\pi^2+ \frac 32 \alpha$; in such a limit the solution is a plane wave function (see Figures \ref {Fig10} and \ref {Fig11}, broken lines). 
%
%
\ If $\alpha < -\Lcost$, then $t^M \not= 0$ and 
  $C_1 \to 0$ as ${\E}\to {\E}^M$; hence, in this limit we have that 
 $\theta (x) \equiv 0$ and the solution has the form 
 \bee
 \Npsi (x) = C {\dn} (qx;t)
 \eee
 already discussed in (\ref {realeTer}) (see Figure \ref {Fig9}, broken line).
 \end {corollary}
 
We consider now the behavior of the solution when $\E$ takes the limit value ${\E}\to {\E}^m$. \ Even in this case the proof of the Corollary below is an immediate 
consequence of Theorem \ref {Teo2}.
 
\begin {corollary} \label {Cor2} If $\alpha <0$, then $A+B\to 0$ and $C_1\to 0$ as $\E \to \E^m$; in such a limit the solution has the form (see Figures \ref {Fig9} and \ref {Fig10}, full lines)
 \bee
 \Npsi (x) = C {\cn} (qx;t) 
 \eee
 already discussed in (\ref {realeBis}). \ If $\alpha >0$, then $B\to 0$ and $C_1 \to 0$ as $\E \to \E^m$; in such a limit the solution has the form (see Figure \ref {Fig11}, full line)
 \bee
 \Npsi (x) = C {\sn} (qx;t) 
 \eee
 already discussed in (\ref {reale}).
 \end {corollary}

 Concerning the quasimomentum we recall that 
\bee
k^m = 0 \ \mbox { if } \ \alpha \le -\Lcost 
\eee
and we can prove that 
\bee
k^M = \pi \ \mbox { if } \ \alpha <0 \ \mbox { and } \ k^m = \pi \ \mbox { if } \ \alpha >0 \, . 
\eee

\begin {theorem} \label {Teo3}
\be
\lim_{\E \to \E^m} k(\E ) = \pi \, . \label {Eqk}
\ee
\end {theorem}

\begin {proof}
In order to compute the quasimomentum $k(t)$ as function of the parameter $t$ we make use of equation (\ref {Eq4BIS}) and we refer to the formula (\ref {terzo}) 
in Appendix obtaining that 
\be
k &=& C_1 \int_0^1 \frac {1}{A\sn^2 (2 K(t) x;t)+B} dx = 2 C_1 \int_0^{1/2} \frac {1}{A\sn^2 (2 K(t) x;t)+B} dx \nonumber \\ 
&=& \frac {\sqrt {(1+A/B)(2\alpha B + 16 K^2(t))}}{2 K(t)} \Pi \left (1; -A/B,t \right ) \label {App2}
\ee

We consider then, at first, the attractive case $\alpha <0$; we have that
\be
k 
&=& \frac {\sqrt {2\alpha  + 16 K(t) {E}(1;t)}}{2 K(t)} \sqrt {1-\zeta }\Pi \left (1; \zeta ,t \right ) \label {App2a}
\ee
where we set $\zeta = -A/B <1$ because of the constraints $A+B>0$ and $B>0$. \ Now, from Corollary \ref {Cor2} it follows that $A+B\to 0$ as $\E \to \E^m$ in the 
attractive case; then we have to compute the following limit
\bee
\lim_{\zeta \to 1-0\, , \ t \to t_2} \sqrt {1-\zeta } \Pi \left (1; \zeta ,t \right ) \, . 
\eee
To this end we observe that $t_2 <1$ and we recall that (see formula (412.01) \cite {BF})
\bee
 \sqrt {1-\zeta } \Pi \left (1; \zeta ,t \right ) 
&=&  \sqrt {1-\zeta }  K(t) + \frac {\pi \sqrt {\zeta } \left ( 1- \Lambda_0 (\varphi , t) \right ) }{2 \sqrt {\zeta -t^2} } \, , 
\eee
when $t^2 < \zeta <1$ and where $\varphi = \sin^{-1} \left ( \sqrt {(1-\zeta )/t'} \right )$, $t'=\sqrt {1-t^2}$ and 
\bee
\Lambda_0 (\varphi , t) = \frac {2}{\pi} \left [ E(t) F(\varphi , t') + K(t) E(\varphi , t') - K(t) F(\varphi ,t' ) \right ] \, , 
\eee
here $E(t)$ denotes the complete elliptic integral of the second kind with parameter $t$ and $F(\varphi , t)$ denotes the normal elliptic integral of first kind with 
argument $\varphi$ and parameter $t$. \ Hence,
\bee
\lim_{\zeta \to 1-0\, , \ t \to t_2} \sqrt {1-\zeta } \, \Pi \left (1; \zeta ,t \right ) 
=   \frac {\pi \left ( 1- \Lambda_0 (0 , t_2) \right ) }{2 \sqrt {1 -t_2^2} } = \frac {\pi  }{2 \sqrt {1 -t_2^2} }
\eee
since $\varphi \to 0$ as $\zeta \to 1$. \ Hence, in such a limit we have that
\bee
k (\E^m) = \frac {\sqrt {2\alpha  + 16 K(t_2) {E}(1;t_2)}}{2 K(t_2)} \cdot \frac {\pi  }{2 \sqrt {1 -t_2^2} } = \pi
\eee
because $t_2$ is such that $8 K^2 (t_2) t_2^2 F_2 (t_2) =-\alpha$.

In order to give the proof in the repulsive case $\alpha >0$ me still make use of formula (\ref {App2}) in the form 
\be
k &=& \frac {\sqrt {2\alpha B + 16 K^2(t)}}{2 K(t)} 
\sqrt {1+\kappa } \Pi \left (1; -\kappa ,t 
\right )  \label {App2b}
\ee
where we set $\kappa = A/B >-1$. \ Now, from Corollary \ref {Cor2} it follows that $B\to 0$ as $\E \to \E^m$ in the repulsive case; 
then we have to compute the following limit for any $t<1$ fixed 
\bee
\lim_{\kappa \to +\infty } \sqrt {1+\kappa } \Pi \left (1; -\kappa ,t \right ) = \lim_{\kappa \to +\infty } \sqrt {1+\kappa } \int_0^1 
\frac {1}{(1+\kappa u^2) \sqrt {1-u^2} \sqrt {1-t^2 u^2 } }du =\frac 12 \pi \, .  
\eee
Hence, in such a limit we have that 
\bee
k (\E^m )= \frac {\sqrt {2\alpha  + 16 K(t_3) {E}(1;t_3)}}{2 K(t_3)} \cdot \frac 12 \pi = \pi
\eee
because $t_3$ is such that $8 K^2 (t_3) t_3^2 F_1 (t_3) =\alpha$.
\end {proof}

\begin {remark}
When $\alpha$ is small enough then $t^M =0$ and $\E^M = \pi^2 + \frac 32 \alpha$. \ Furthermore, in the limit $\alpha \to 0$ a straightforward calculus gives that 
\bee
t_2 := t_2 (\alpha ) \sim \sqrt {-\frac {2\alpha }{\pi }}
\ \mbox { and } \ 
t_3 := t_3 (\alpha ) \sim \sqrt {\frac {2\alpha }{\pi }} \, .
\eee
In such a case the limit of the solution as $\alpha$ goes to zero becomes the plane wave solution discussed in Remark \ref {Nota1} associated to $k=\pi$.
\end {remark}

\appendix

\section {Some formulas concerning Jacobian Elliptic functions} \label {AppA}

In order to give an explicit expression to the function $F_1(t)$ we recall that (see Formula 310.02 \cite {BF})
\be
F_1(t) = 2   \int_0^{1/2} {\sn}^2 (qx;t) dx = \frac {q-2 E \left (  {\sn} (q/2;t);t \right )}{qt^2} 
= \frac {q-2 E \left ( 1;t \right )}{qt^2} \, , \label {App1}
\ee
since $\sn (q/2;t) =1$ when $q=2 K(t)$, where $E (\varphi ;t)$ is the incomplete elliptic integral of second kind with argument $\varphi$ and parameter $t$. 

Here, we collect some results concerning the elliptic integral of third kind defined as 
\bee
\Pi  (z;\nu , t) = \int_0^z \frac {1}{(1-\nu u^2) \sqrt {1-u^2} \sqrt {1-t^2 u^2 } }du \, . 
\eee
Recalling that $\frac {d }{dx} \sn (x;t) = \cn (x;t) \dn (x;t)$ and that $\sqrt {1-\sn^2 (x;t)} = |\cn (x;t)|$ and $\sqrt {1-t^2\sn^2 (x;t)} = \dn (x;t)$ then 
\be
\int \frac {1}{A\sn^2 (x;t)+B} dx = \frac {1}{B} \Pi \left (\sn (x;t); -A/B,t \right ) \label {terzo}
\ee
provided that $x \in [0,K(t)]$ because $|\cn (x;t)| = \cn (x;t)$, $B\not= 0$, $-A/B < 1$ and $t\in [0,1)$.

\end {document}